\documentclass[10pt,a4paper,reqno]{amsart}
\usepackage{amsmath}
\usepackage{amsfonts}
\usepackage{amssymb}
\usepackage{graphicx}

\usepackage{anysize}
\marginsize{3cm}{3cm}{3cm}{3cm}

\newtheorem{theorem}{Theorem}[section]

\begin{document}
\setcounter{page}{1}
\bigskip
\title[\centerline{The contraction of the representations of the group $SO(4,1)$}
\hspace{0.5cm}] {The contraction of the representations of the group $SO(4,1)$ and cosmological interpretation}

\author[\hspace{0.7cm}\centerline{The contraction of the representations of the group $SO(4,1)$}]{B.A. RAJABOV$^1$ }

\thanks {\noindent $^1$ N.Tusi Shamakhi Astrophysics Observatory, National Academy of Sciences of Azerbaijan;\\
\indent \,\,\, e-mail: balaali.rajabov@mail.ru.}

\bigskip
\begin{abstract}
	We consider the operation of contraction of unitary irreducible representations of the de Sitter group $ SO(4,1) $. It is shown that a direct sum of unitary irreducible representations of the Poincar\'{e} group with different signs of the rest mass is obtained as a result of contraction. The results obtained are used to interpret the phenomena of "dark matter" and "dark energy" in terms of the elementary quantum systems of de Sitter's world.
	

\end{abstract}

\maketitle
\bigskip

	\section{Introduction}
	In this article, we will consider the application of the operation of contraction for representations of the group of movements of the de Sitter world. The operation of contraction was introduced in the paper of Wigner-In\"{o}n\"{u} for the Lie algebra of the relativistic Poincar\'{e} group, and the speed of light was used as the contraction parameter, \cite{Wigner-Inonu-1953}. As a result, a transition to the non-relativistic Galilean group was obtained. Subsequently, this operation was generalized to all Lie algebras, \cite{Barut-Raczka}.
	
	The contraction of the representations of the de Sitter group $SO(4,1)$ from the points of Lie algebra was considered in \cite{Strom1961}-\cite{Strom1965}. In the article \cite{Mic-Nied}, the contraction operation is extended to representations of the continuous basic series of groups $ SO(p,1) $, and two types of contraction are considered: to representations of the non-homogeneous Lorentz group $ ISO(3,1) $ and to the representations of the group of Euclidean motions $M(p)$. Drechsler considered contraction in a fiber bundle with Cartan connection, in particular, with the de Sitter group bundle, \cite{Drechsler}. 
	
	In \cite{Rajab}, explicit expressions for the matrix elements of non-degenerate irreducible representations of the group $ M(4) $ are found using the operation of contraction of the matrix elements of unitary irreducible representations of the continuous fundamental series of $ SO(4,1) $ group.
	
	The monograph \cite{Gromov} is devoted to multidimensional generalizations of contraction for classical and quantum groups.
	
	The papers \cite{Nachtmann1967}-\cite{Borner-Durr} consider the quantum field theory in the de Sitter world. The problems of dark matter and dark energy are reviewed \cite{Ryabov}-\cite{DarkMatter-3}.
	
	The present paper is devoted to the contraction of unitary irreducible representations of the de Sitter group $SO(4,1)$. The scalar radius of curvature of space-time is chosen as the contraction parameter. Using the of Wigner-In\"{o}n\"{u} limit, a direct sum of unitary irreducible representations of the Poincar\'{e} group with different energy signs is obtained. Results obtained are used to interpret "dark matter" and "dark energy" using the elementary systems of de Sitter's world.
		
	The main results of this work were reported in the 3rd international Scientific Conference "Modern Problems of Astrophysics-III", September 25-27, 2017, Georgia.
	
	\section{Background ... Spin is a purely quantum phenomenon}	
	
	The Stern-Gerlach experiment proved the existence of an intrinsic orbital angular momentum and the magnetic moment of the electron, which assume discrete values (1922), \cite{Stern}.
	
	The works of Goudsmit and Uhlenbeck showed that an attempt to explain the spin in the representations of classical physics leads to insuperable difficulties (appear speeds greater than the speed of light, 1925), \cite{Goudsmit}-\cite{Thomas}.
	
	V. Pauli proposed a non-relativistic Schr\"{o}dinger equation for an electron with allowance for the spin variable (Pauli equation for a 2-component electron, 1927), \cite{Schiff}.
	
	P. Dirac developed the relativistic theory of the electron and introduced 4-component wave functions (Dirac equation, 1928), \cite{Bjorken}.
	
	The main conclusion was that the electron spin is a purely quantum phenomenon. The peak of these theories was the connection of spin with statistics and the prediction of the positron.
	
	\textit{Since then, in all reviews, monographs, and courses, this story is given  to explain the quantum nature of spin...}

	\section{Wigner's theory of elementary systems}
	
	A real understanding of the nature of the spin appeared after the classical work of E. Wigner on the representations of the inhomogeneous Lorentz group (that is, the Poincar\'{e} group, 1937), \cite{Wigner1937}. He established that projective representation of this group acts in the Hilbert space of states. Wigner introduced the concept of an elementary system and showed that such systems correspond to irreducible representations, that is, the space of these representations does not have invariant subspaces.
	
	Most importantly. Wigner proved that the rest mass and the spin of an elementary particle (!) are:
	\begin{itemize}
		\item  invariants that uniquely characterize irreducible representations of the Poincar\'{e} group;
		\item consequences of the space-time symmetries and are not consequences of the equations of dynamics.\footnote{Thus, the justification of the spin by the Pauli and Dirac equations is of historical significance. Rephrasing Oscar Wilde, we can say that if Wigner's work had appeared earlier, the history of quantum electrodynamics would have evolved differently.}
	\end{itemize}
	 
	Invariants of the inhomogeneous Poincar\'{e} group (Casimir operators) are constructed using translation generators $ P_{\mu} $ and 4-dimensional rotations $ M_{\mu\nu} $ as follows:
	\begin{equation}
	m^{2}=P^{\mu}P_{\mu};\qquad w^{2}=w^{\mu}w_{\mu}=m^{2}s(s+1);\qquad w_{\varrho}=\dfrac{1}{2}\epsilon _{\lambda\mu\nu\varrho}P^{\lambda}M^{\mu\nu},
	\end{equation}
	where
	
	$ m $ -- mass, $ s $ -- spin of the particle, and $ w_{\varrho} $ is the Pauli-Lubansky-Bargmann vector.
	
	\textbf{\textit{Important!}} In the case $ m^{2}\geq 0 $, there is a third invariant - the sign of energy:
	\[ \varepsilon = \dfrac{P_{0}}{|P_{0}|}=\pm 1 \]
	
	\textit{So, spin is a purely quantum phenomenon, but it is associated with groups of space-time movements ...}\footnote{In 1947, Wigner and Bargmann, on the basis of the theory of representations of the Poincar\'{e} group, classified the relativistic equations for spin particles, \cite{Bargmann}. That is, first the spin, and then the equations ...}
	
	\section{Einstein's equations and de Sitter's solutions}
	According to Einstein's general relativity, the metric properties of space-time are determined by the distribution and motion of matter, \cite{Einstein}:
	\begin{equation}
	R_{\mu\nu}-\dfrac{1}{2}g_{\mu\nu}R+\varLambda g_{\mu\nu}=-\dfrac{8\pi G}{c^{4}}T_{\mu\nu};\quad R=R_\mu^\mu
	\end{equation}
where

$ R_{\mu\nu} $ -- Ricci tensor, $ g_{\mu\nu} $ -- metric tensor, $ \varLambda $ -- cosmological constant, $ G $ -- the gravitational constant of Newton, $ T_{\mu\nu} $ -- energy-momentum tensor and $ \mu,\nu=0,1,2,3. $.

This equation can be rewritten in an equivalent form:
\begin{equation}
R_{\mu\nu}=\varLambda g_{\mu\nu}-\dfrac{8\pi G}{c^{4}}(T_{\mu\nu}-\dfrac{1}{2}g_{\mu\nu}T);\quad T= T_\mu^\mu
\end{equation}
It is clear from (2) that the $ \varLambda $-term, even in the absence of matter ($ T_{\mu\nu}=0 $), changes the space-time geometry and $ g_{kl}\Lambda $ is the energy-momentum tensor of the vacuum.\footnote{The values of $ G $ and $ \varLambda $ are constantly refined with the accumulation of observations:
	
	$ G=6,67545\frac{sm^{3}}{g \cdot sec^{2}} $,\quad (2013);\qquad $ \Lambda \sim 10^{-53}m^{-2},\quad (1998) $.}  

In the vacuum, the Einstein equations take the form:
\begin{equation}
R_{\mu\nu} = \varLambda g_{\mu\nu}.
\end{equation}
For $ \varLambda = 0 $, the solution of (3) is the Minkowski manifold with a group of Poincar\'{e} motions.

The histories of the cosmological constant are reviewed in \cite{Straumann}-\cite{One Hundred}.

In the general case, the solutions of the Einstein (1)-(2) equations do not have a group of motions. But in 1917 Willem de Sitter found two solutions of (3) for $ \Lambda\neq 0 $ with different global groups of movements, \cite{Weinberg}--\cite{Tolman}: 
\begin{equation}
ds^{2}=\dfrac{dr^{2}}{1-r^{2}/R^{2}}+r^{2}(d\vartheta^{2}+\sin^{2}\vartheta d\varphi^{2})-\left({1-\dfrac{{r^{2}}}{{R^{2}}}}\right)c^{2}dt^{2},\qquad if \ \Lambda > 0;
\end{equation} 
\begin{equation}
ds^{2}=\dfrac{dr^{2}}{1+r^{2}/R^{2}}+r^{2}(d\vartheta^{2}+\sin^{2}\vartheta d\varphi^{2})-\left({1+\frac{r^{2}}{R^{2}}}\right)c^{2}dt^{2},\qquad if \ \Lambda < 0.
\end{equation} 
Here the radius of curvature $ R $ and the cosmological constant $ \Lambda $ are related by the following formula:\footnote{It is easy to see that for $ \Lambda\rightarrow 0 $, i.e. $ R\rightarrow\infty $ both solutions (4)-(5) are transferred to the flat world of Minkowski.}
\[ \Lambda = \pm\dfrac{3}{R^{2}} \]

Using stereographic projections:
\begin{eqnarray*}
\xi_{1}=r\cos \varphi;\quad \xi_{2}=r\sin\vartheta\cos\varphi;\quad \xi_{3}=r\sin\vartheta\sin\varphi;\\\xi_{4}=R\sqrt{1-\dfrac{r^{2}}{R^{2}}}\cosh\left(\dfrac{ct}{R}\right);\quad \xi_{0}=R\sqrt{1-\dfrac{r^{2}}{R^{2}}}\sinh\left(\dfrac{ct}{R}\right)
\end{eqnarray*}
and 
\begin{eqnarray*}
	\eta_{1}=r\cos \varphi;\quad \eta_{2}=r\sin\vartheta\cos\varphi;\quad \eta_{3}=r\sin\vartheta\sin\varphi;\\\eta_{4}=R\sqrt{1+\dfrac{r^{2}}{R^{2}}}\cosh\left(\dfrac{ct}{R}\right);\quad \eta_{5}=R\sqrt{1+\dfrac{r^{2}}{R^{2}}}\sinh\left(\dfrac{ct}{R}\right)
\end{eqnarray*}
de Sitter solutions can be isometrically embedded as sub-manifolds in 5-dimensional pseudo-Euclidean spaces:
\begin{equation}
\xi_{0}^{2}-\xi_{1}^{2}-\xi_{2}^{2}-\xi_{3}^{2}-\xi_{4}^{2}=-R^{2}\label{eq:07}
\end{equation}
and
\begin{equation}
\eta_{1}^{2}+\eta_{2}^{2}+\eta_{3}^{2}-\eta_{4}^{2}-\eta_{5}^{2}=-R^{2},\label{eq:8}
\end{equation}
respectively.

These spaces have global symmetry groups $ SO(4,1) $ and $ SO(3,2) $ that leave the metrics (6)-(7) invariant. Groups $ SO(4,1) $ and $ SO(3,2) $ are called de Sitter groups. Spaces (4),(6) and (5),(7) are called de Sitter worlds of the 1st and 2nd kind or according to the modern terminology, de Sitter worlds $dS$ and anti--de Sitter $AdS.$

We restrict ourselves to the de Sitter world (5),(7) and the $ SO(4,1) $ group. The case of the anti-de Sitter world will be considered in a separate paper, because of the difficulties in interpreting the space-time coordinates.

The commutation relations for the generators of the group $ SO(4,1) $ have the form:
\begin{eqnarray*}
\left[ M_{\mu\nu},M_{\varrho\sigma} \right] =i\left( g_{\mu\varrho}M_{\nu\sigma}-g_{\mu\sigma}M_{\nu\varrho}-g_{\nu\varrho}M_{\mu\sigma}+g_{\nu\sigma}M_{\mu\varrho} \right);\\
\left[ M_{\mu\nu},P_{\varrho}\right] = i\left( g_{\mu\varrho}P_{\nu}-g_{\nu\varrho}P_{\mu}\right);\qquad \quad \left[ P_{\mu},P_{\nu}\right] =\dfrac{i}{R^{2}}M_{\mu\nu}; 
\end{eqnarray*}
where $ P_\mu=(1/R)M_{4\mu} $.

The Casimir operators of the Lie algebra of the group $ SO(4,1) $ have the following form, \cite{Gursey}:
\begin{equation}
C_{1}=-\dfrac{1}{2R^{2}}M_{ab}M^{ab}=-P_{\lambda}P^{\lambda}-\dfrac{1}{2R^{2}}M_{\mu\nu}M^{\mu\nu}=M^{2},\label{eq:9}
\end{equation}
and
\begin{equation}
C_{2}=-W_{a}W^{a},\quad W_{a}=\dfrac{1}{8R}\varepsilon_{abcde}M^{bc}M^{de}.\label{eq:10}
\end{equation}

	Here, the lifting and lowering of the indexes are carried out using the 5-dimensional metric tensor \eqref{eq:7}.
	
	To consider the result contraction $ R\rightarrow\infty $, it is convenient to represent the operator $ C_{2} $ in the following form:
	\begin{equation}
	C_{2}=-V_{\lambda}V^{\lambda}-\dfrac{1}{R^{2}}W_{4}^{2},
	\end{equation}
	where 
	\begin{eqnarray}
	W_{\lambda}&=&-\dfrac{1}{2}\varepsilon_{\lambda\varrho\mu\nu 4}P^{\varrho}M^{\mu\nu},\\
	W_{5}&=&\dfrac{1}{8}\varepsilon_{\lambda\mu\nu\varrho}M^{\lambda\mu}M^{\nu\varrho}.
	\end{eqnarray}

	From the last expressions, it is clear that when $ R\rightarrow \infty $ the Lie algebra of the group $ SO(4,1) $ becomes over to the Lie algebra of the Poincar\'{e} group, and the Casimir operators become:
	\begin{eqnarray}
	C_{1}&\rightarrow&-P_{\lambda}P^{\lambda}= m^{2},\\
	C_{2}&\rightarrow&-V_{\lambda}V^{\lambda}=m^{2}s(s+1).	
	\end{eqnarray}
	where $ s,m $ -- spin and rest mass, respectively.
	
	From the limiting transition of the Casimir operators $ C_{1},C_{2} $ it follows that unlike the Minkowski world in the de Sitter world, elementary systems are identified not by mass and spin, but by some functions of spin and mass.

	\textbf{\textit{Important!}} It is obvious that any function of invariants is invariant. Therefore, as invariants characterizing unitary irreducible representations of the group $ SO(4,1) $, any pair of functionally independent invariants can be chosen.	
	
	In particular, non-degenerate representations of the group $ SO(4,1) $ can be realized in the space of $ (2s+1) $-component vector-functions and the degree of homogeneity of $ \sigma $ on the upper field of the cone, \cite{Rajab, Rajab-2,Rajab-3}. Then the parameters $ s $ and  $ \sigma $ will play the role of invariants characterizing the irreducible representations of the group $SO(4,1)$. 
	
	In addition, as will be shown in the next section, after the operation of contraction, the parameter $ s $ becomes into spin and the $ \sigma $ to the function of mass $ m $.
	
	\section{The contraction of the representations of the group $ SO(4,1) $}
	Representations of the group $ SO(4,1) $ will be constructed in the space of measurable vector-valued functions $\Phi:SO(4,1)\longrightarrow C^{2s+1}$ that satisfies the following conditions:
	\begin{equation}
	\Phi_{\lambda}(gp)=\sum_{\lambda'=-s}^{s}\Delta_{\lambda\lambda'}^{(\sigma,s)}(p^{-1})\Phi_{\lambda'}(g),\qquad Re\,\sigma=-3/2;\label{eq:185}
	\end{equation}
	\[
	\sum_{\lambda=-s}^{s}\int\left|\Phi_{\lambda}(p)\right|^{2}\varrho(g)dg<\infty.
	\]
	Here, $ \rho (\cdot) $ is a continuous non-negative function on the group $ SO (4,1) $ such that
	\begin{equation}
	\int\varrho(gp)d_{l}(p)=1\label{eq:186}
	\end{equation}
	and supp$\,\rho$ has a compact intersection with each class of the contiguity $ gP $. Such a function exists for any locally compact group \cite{Kirillov}. We introduce in this space a scalar product:
	\begin{equation}
	\left(\Phi^{(1)},\Phi^{(2)}\right)=\sum_{\lambda=-s}^{s}\overline{\Phi_{\lambda}^{(1)}(g)}\Phi_{\lambda}^{(2)}(g)\varrho(g)dg.\label{eq:187}
	\end{equation}
	
	The Hilbert space obtained in this way will be denoted by $ \mathfrak{L}^2\left( SO(4,1); \Delta ^{(\sigma,S)}\right) $ and consider the representation of the group $ SO (4,1) $ acting in this space along formula:
	\begin{equation}
	T^{(\sigma,s)}(g)\Phi_{\lambda}(g_{1})=\Phi_{\lambda}(g^{-1}g_{1}).\label{eq:188}
	\end{equation}
	
	This is the induced representation in the Mackey sense of the group  $ SO (4,1) $, \cite{Mackey}.
	
	We now construct another realization of the space used for studying problems of contraction of representations. For this, we consider on the cone:
	\[ \left[ k,k\right] = k_{0}^{2}-k_{1}^{2}+k_{2}^{2}+ k_{3}^{2}+k_{4}^{2}=0,\quad k_{0}>0,\]
	the new coordinate system:
	\begin{eqnarray}
	k & = & \left|k_{4}\right|(\zeta,\varepsilon),\nonumber \\
	\zeta_{0} & = & \cosh\beta,\nonumber \\
	\zeta_{1} & = & \sinh\beta\sin\vartheta\sin\varphi,\label{eq:191}\\
	\zeta_{2} & = & \sinh\beta\sin\vartheta\cos\varphi,\nonumber \\
	\zeta_{3} & = & \cos\vartheta\sinh\beta;\nonumber 
	\end{eqnarray}
	\[
	\left|k_{4}\right|\neq0,\quad\varepsilon=\pm1,\quad0\leq\beta<\infty,\quad0\leq\vartheta\leq\pi,\quad0\leq\varphi<2\pi.
	\]
	\begin{equation}
	\zeta_{0}^{2}-\zeta_{1}^{2}-\zeta_{2}^{2}-\zeta_{3}^{2}=1.\label{eq:192}
	\end{equation}
	
	The set of points of the upper half of the cone that is not covered by this coordinate
	system forms a manifold of lower dimension. Since the upper half of the cone is a transitive surface, one can obtain parametrization of the elements of the group starting from this coordinate system.
	
	First, we fix one of the points of the cone, namely, the point $ \mathring{k} $:
	\begin{equation}
	\left|k_{4}\right|=1,\quad\varepsilon=1,\quad\beta=\vartheta=\varphi=0.	
	\end{equation}
Next we note that the stationary group $\mathfrak{W} $ of the point $ \mathring{k} $ is isomorphic to the group motions $M(3)$ of 3-dimensional Euclidean space $ E_ {3} $. Each element $w\in \mathfrak{W}$ can be uniquely represented as:
\begin{equation}
w=R(\rho)B(\vec{z}),\label{eq:16}
\end{equation}
where $ \vec {z} $ is a 3-dimensional vector, and $ \rho $ is an orthogonal matrix $ 3 \times 3 $. The matrices $ R (\rho) $ and $ B (z) $ have the following form:
	\begin{equation}
	R(\rho)=\begin{pmatrix}1 & \vec{0\,}^{\intercal} & 0\\
	\vec{0} & \rho & \vec{0}\\
	0 & \vec{0\,}^{\intercal} & 1
	\end{pmatrix},\quad R(\rho)\in SO(3)\subset SO(4,1),\label{eq:17}
	\end{equation}
	\begin{equation}
	B(\vec{z})=\begin{pmatrix}1+z^{2}/2 & \vec{z\,}^{\intercal} & -z^{2}/2\\
	\vec{z} & I_{3} & -\vec{z}\\
	z^{2}/2 & \vec{z\,}^{\intercal} & 1-z^{2}/2
	\end{pmatrix},\quad z^{2}=\vec{z\,}^{2}\label{eq:18}
	\end{equation}
	The matrices $ B (\vec {z}) $ form an abelian subgroup that is normal division subgroup of the stationary subgroup $ \mathfrak{W}. $
	
	The transitivity property allows an almost element $ g \in SO(4,1) $ to represent an expansion:
	\begin{equation}
	g=h_{\varepsilon}(k)w,\quad\varepsilon=\pm 1,\quad w \in \mathfrak{W},\label{eq:195}
	\end{equation}
	where $ h_{\varepsilon}(k) $ is the so-called Wigner operator ("boost") having the property:
	\[ k=h_{\varepsilon}(k)\mathring{k}. \]
	The Wigner operator is defined by the following formulas:
	\begin{equation}
	h_{1}(k)=g_{12}(\varphi)g_{23}(\vartheta)g_{03}(\beta)g_{04}(\tau),\qquad\tau=\ln\left|k_{4}\right|,\label{eq:193}
	\end{equation}
	for the subset $ \varepsilon = 1 $, and
	\begin{equation}
	h_{-1}(k)=g_{12}(\varphi)g_{23}(\vartheta)g_{03}(\beta)\eta g_{04}(\tau),\qquad\tau=\ln\left|k_{4}\right|,\label{eq:194}
	\end{equation}
	for the subset $ \varepsilon = -1 $.
	
	Here $ \eta $ is the diagonal $ 5\times 5 $ matrix: $ \eta = 
	\mathrm{diag}(1,-1,-1,-1,-1) $.
	
	We now note that the action of $ g \in SO (4,1) $ on a cone induces nonlinear transformation of this element on the sections $\left | k_ {4} \right | = 1 $, which are the uppers of a two-sheeted hyperboloid.
	
	It follows from \eqref{eq:193}-\eqref{eq:195} that the stationary subgroup of this variety is the minimal parabolic subgroup $ \mathsf {P} $, \eqref{eq:31} and for almost all $ g \in SO(4,1) $ the expansion is valid:
	\begin{equation}
	g=h_{\varepsilon}(\zeta)p,\label{eq:196}
	\end{equation}	
	where
	\[
	h_{\varepsilon}(\zeta)=h_{\varepsilon}(k)|_{\tau=0},\quad\varepsilon=\pm1.
	\]
	
	Let $ \mathfrak{H}_ {i}, \: i = 1,2 $ be the Hilbert spaces of vector functions on the hyperboloid \eqref{eq:192} with the scalar product:
	\begin{equation}
	(F^{(1)},F^{(2)})=\int\overline{F^{(1)}(\zeta)}F^{(2)}(\zeta)d\mu(\zeta),\label{eq:197}
	\end{equation}
	where
	\begin{eqnarray}
	d\mu & = & \frac{\left(d\vec{\zeta}\right)}{2\sqrt{1+\zeta^{2}}},\label{eq:198}\\
	\left(d\vec{\zeta}\right) & = & d\zeta_{1}d\zeta_{2}d\zeta_{2},\nonumber 
	\end{eqnarray}
	is the usual Lorentz-invariant measure on the field of a two-sheeted hyperboloid.
	
	We define an isometric mapping:
	\[
	\mathfrak{L}^{2}\left(SO(4,1);\Delta^{(\sigma,s)}\right)\longrightarrow\mathfrak{H}_{1}\oplus\mathfrak{H}_{2}
	\]
	in the following way:
	\begin{equation}
	F_{\lambda}(\zeta;\varepsilon)=\Phi_{\lambda}(h_{\varepsilon}(\zeta)).\label{eq:199}
	\end{equation}
	The inverse mapping is found with the help of \eqref{eq:185} and \eqref{eq:195}:
	\begin{equation}
	\Phi_{\lambda}(g)=\sum_{\lambda'=-s}^{s}\Delta_{\lambda\lambda'}^{(\sigma,s)}(p^{-1})F_{\lambda'}(\zeta;\varepsilon).\label{eq:200}
	\end{equation}
	
	These mapping properties are easily verified. It follows from \eqref{eq:185}, that under the mapping \eqref{eq:199}-\eqref{eq:200} the representation of the group $ SO(4,1) $ in the space $ \mathfrak{L}^{2} \left (SO(4,1); \Delta^{(\sigma, s)} \right) $ becomes a unitary representation of this group in $ \mathfrak{H}_{1} \oplus \mathfrak{H}_{2} $:
	\begin{equation}
	T^{(\sigma,s)}(g)F_{\lambda}(\zeta;\varepsilon)=\sum_{\lambda'=-s}^{s}\Delta_{\lambda\lambda'}^{(\sigma,s)}(p_{g}^{-1})F_{\lambda'}(\zeta_{g};\varepsilon_{g}),\label{eq:201}
	\end{equation}
	where $ \zeta_{g}, \varepsilon_{g}, p_{g} $ - are determined by the equation:
	\begin{equation}
	g^{-1}h_{\varepsilon}(\zeta)=h_{\varepsilon_{g}}(\zeta_{g})p_{g}.\label{eq:202}
	\end{equation}
	Thus, we have found the realization of the $ UIR $'s of the continuous main series of the group $ SO(4,1) $ in the space $ \mathfrak{H}_{1} \oplus \mathfrak{H}_{2} $.
	
	Consider now the Wigner-Inonu limit at which the $ UIR $'s of the group $ SO(4,1) $ goes into the $ UIR $'s of the group $ ISO(3,1) $. For this, it is necessary to know the limit of the parameters of the group $ SO(4,1) $, under which the de Sitter group passes to the Poincar\'{e} group $ ISO(3,1) $.
	
	Since the  $ SO(4,1) $ is a group of motions leaving an invariant quadratic form \eqref{eq:07} It is clear that to find the required limit it is sufficient to consider
	only transformations in the planes $ (\xi_{0}, \xi_{j}), \; j = 1,2,3,4 $. For example, consider the transformation:
	\[
	\xi_{0}^{'}=\xi_{0}\cosh\alpha+\xi_{4}\sinh\alpha,\qquad \xi_{4}^{'}=\xi_{0}\sinh\alpha+\xi_{4}\cosh\alpha.
	\]
	Since the de Sitter space \eqref{eq:07} when $ R \rightarrow \infty $ transforms to Minkowski space, we redefine parameter $ \alpha $, putting $ \alpha = a_ {0} /R $. Since
	\begin{equation}
	\xi_{4}=\sqrt{\xi_{0}^{2}-\xi_{1}^{2}-\xi_{2}^{2}-\xi_{3}^{2}+R^{2}},\label{eq:3119}
	\end{equation} 
	we obtain:
	\begin{eqnarray*}
		\xi_{0}' & = & \xi_{0}\cosh\left(\frac{a_{0}}{R}\right)+\sqrt{\xi_{0}^{2}-\xi_{1}^{2}-\xi_{2}^{2}-\xi_{3}^{2}+R^{2}}\sinh\left(\frac{a_{0}}{R}\right)\xrightarrow{R\rightarrow\infty}\xi_{0}+a_{0},\\
		\xi_{4}' & \rightarrow & \xi_{4}.
	\end{eqnarray*}
	Thus, we see that the transformations $ g_{04} \left (\alpha \right), \, g_{i4} \left (a_{i} \right), \, i = 1,2,3 $ in the planes $ \left (\xi_{n}, \, \xi_{4} \right), \, n = 0,1,2,3 $ go to the translation group if we assume:
	\[
	\alpha=a_{0}/R,\;\varphi_{i}=a_{i}/R,\;i=1,2,3\qquad\text{if}\quad R\rightarrow\infty.
	\]
	
	Thus we have the statement: the transformations of the subgroup $ SO(3,1) \subset SO(4,1) $, leaving the vertex of the one-sheeted hyperboloid $ [k, k] = -R^{2} $ (i.e. the de Sitter world) transform into homogeneous Lorentz transformations of the Minkowski world, and the rotation $ g_ {n4} (\alpha_ {n}) $ on the planes $ (k_{4}, k_{n}), \; n = 0,1,2,3 $ translate into translations along the corresponding axes of a Cartesian coordinate system on the quantities: $ a_{n} = \alpha_{n} R $.
	
	The result of contraction of the representation \eqref{eq:201}-\eqref{eq:202} of the group $SO(4,1)$ is formulated as a theorem:
	\begin{theorem}	
			 The result of the contraction of the UIR's, $ T ^ {(\sigma, s)} (g), \: g \in SO (4,1), \sigma = -3/2 + imR, $ for $ R \rightarrow \infty $ is the direct sum of UIR, $ U ^ {(m, s; \varepsilon)} (g), \: g \in ISO (3,1) $, with mass $ m $, spin $ s $ and differing in energy sign $ \varepsilon $.			
	\end{theorem}
	\begin{proof}
		The proof of the theorem is given in the full version of the article, \cite{Caucasus-3}.	
	\end{proof}
	This conclusion is the same as the result of \cite {Strom1965}, obtained using infinitesimal transformations.
	
	\section{Conclusion}
	So let's sum up ...
	
    In the flat world of Minkowski, Wigner's elementary systems are determined by the rest mass, the spin and the sign of their energy can be identified with elementary particles.
    Considerations of the stability of physical systems force us to limit the spectrum of energy from below and to exclude negative energies.
	
	In the de Sitter world, elementary Wigner systems are identified by spin and by a parameter, which is the flat limit of a function of spin and mass, with different energy signs. But unlike the Minkowski world, we can not exclude negative energies from consideration. 
	
	That is, elementary systems on a cosmological scale can be in states with positive and negative energies. Elementary systems in a state with positive energy behave like a gravitating mass, and in a negative energy state as an anti-gravity mass.
	
	Thus, we arrive at the following conclusions:
	\begin{enumerate}
		\item Mysterious "dark matter" and "dark energy" consist of such elementary systems.
		\item "Dark matter" and "dark energy" are the first manifestations of quantum properties on the scale of the universe. Until now, quantum phenomena have been encountered in the micro-world, and also as macroscopic quantum effects in the theory of condensed matter.
		\item "Dark matter" and "dark energy" are carriers of information about the first moments of the universe after the Big Bang.
	\end{enumerate}
	
	The last conclusion follows from the fact that according to the standard cosmological model, the de Sitter world is a necessary phase of the evolution of the universe in the first instants ($ 10^{-34}-10^{-32} $ seconds) after the Big Bang.
	
	Of course, our universe is not de Sitter's world, although according to some data it is developing in the direction of this model. Considering the given phenomena of dark matter and dark energy in the general case is a difficult task because today there is no quantum theory of gravity. The solution of this problem in general for the gravitational field requires not only new physical
	concepts but also new mathematics.
	
		\section*{Appendix: Parametrization of elements of the group $ SO(4,1) $}
		The group $ SO (4,1) $ is a connected component of the unit group of motions
		5-dimensional pseudo-Euclidean space that leaves invariant quadratic form, \cite{Gursey}:
		\[
		\left[k,k\right]=k_{0}^{2}-k_{1}^{2}-k_{2}^{2}-k_{3}^{2}-k_{4}^{2}.
		\]
		It is a 10-parametric, as well as the Poincar\'{e} group, which describes symmetry of Minkowski space.
		
		Elements of $ g \in SO (4,1) $ are represented by $ 5 \times5 $ by matrices that
		linear transformation in the space of points $ k = \left (k_ {0}, k_ {1}, k_ {2}, k_ {3}, k_ {4} \right) $. It follows from the definition that they satisfy the following relations:
		\begin{equation}
		g^{\intercal}\eta g=\eta,\quad\det g=1,\quad g_{00}\geqslant1,\label{eq:6}
		\end{equation}
		where the symbol $ \intercal $ denotes transposition and $ \eta $ is the diagonal $ 5\times 5 $ matrix:
		\begin{equation}
		\eta = \mathrm{diag}(1,-1,-1,-1,-1).\label{eq:7}
		\end{equation}
		
		The matrix elements will be indexed:
		\begin{eqnarray*}
			a,b,c & = & 0,1,2,3,4;\\
			i,j,k & = & 0,1,2,3;\\
			\alpha,\beta,\gamma & = & 1,2,3.
		\end{eqnarray*}
		In the latter case, vector notation will also be used:
		\[
		\vec{a}=\left\{ a^{\alpha}:\:\alpha=1,2,3\right\} .
		\]
		
		Elements of the Lie algebra $ SO (4,1) $ of the de Sitter group $ \Gamma_ {ab} $ are
		five-row matrices with elements of the form:
		\[
		\left(\Gamma_{ab}\right)_{d}^{c}=\delta_{\:a}^{c}\eta_{bd}-\delta_{\:b}^{c}\eta_{ad}
		\]
		and satisfy the commutation relations:
		\[
		\left[\Gamma_{ab},\Gamma_{cd}\right]=\eta_{ac}\Gamma_{bd}-\eta_{bc}\Gamma_{ad}-\eta_{ad}\Gamma_{bc}+\eta_{bd}\Gamma_{ac}
		\]
		Here $ \delta _ {\: a} ^ {b} $ is the Kronecker symbol:
		\[ 
		\delta _ {\: a} ^ {b}=\begin{cases}1, \quad \text{if}\quad   a=b;\\
		0,\quad \text{if} \quad  a\neq b.
		\end{cases}
		\]
		
		Since the transformations $ g \in SO (4,1) $ preserve the form $ \left [k, k \right] = k_ {0} ^ {2} -k_ {1} ^ {2} -k_ {2} ^ { 2} -k_ {3}^{2} -k_ {4} ^ {3} $,  then the surfaces are:	 
		\[
		k_{0}^{2}-k_{1}^{2}-k_{2}^{2}-k_{3}^{3}-k_{4}^{2}=const
		\]
		transform into themselves under transformations from the de Sitter group. 
		
		There are three types of such surfaces, namely the upper (or lower) floor of the two-sheeted hyperboloid $ \left [k, k \right] = c> 0 $, the one-sheeted hyperboloid $ \left [k, k \right] = c <0 $, and finally the upper (or lower) floor of the cone $ \left [k, k \right] = 0 $ without a vertex, since the point $ k_ {0} = k_ {1} = k_ {2} = k_ {3} = k_ {4} = 0 $ itself forms a homogeneous space.
		
		The action of the group on each of these surfaces is transitive. As is known, transitivity surfaces are characterized by their stationary subgroup, i.e. closed subgroup leaving the selected point fixed. Used on the different coordinate systems given on these surfaces, one can obtain different parameterizations of the elements of the group.
		
		We will consider a spherical coordinate system on a cone $[k,k]=0,\;k_{0}>0:$
		\[
		k=\omega(1,\upsilon),
		\]
		where $ \upsilon $ is the point of a 3-dimensional unit sphere:
		\begin{equation}
		\upsilon^{2}=\upsilon_{1}^{2}+\upsilon_{2}^{2}+\upsilon_{3}^{2}+\upsilon_{4}^{2}=1,\label{eq:14}
		\end{equation}
		\[ \upsilon =(\sin\chi\sin\vartheta\sin\varphi, \sin\chi\sin\vartheta\cos\varphi, \sin\chi\cos\vartheta, \cos\chi), \]
		where $0<\omega<\infty,\quad0\leq\chi\leq\pi,\quad0\leq\vartheta\leq\pi,\quad0\leq\varphi<2\pi.$
		
		Coordinates of the selected point $\overset{0}{k}=(1,0,0,0,1):\qquad\omega=1,\;\chi=\vartheta=\varphi=0.$
		
		Wigner's operator:
		\begin{equation}
		h(k)=g_{12}(\varphi)g_{23}(\vartheta)g_{34}(\chi)g_{04}(\beta),\quad\beta=\ln\omega\label{eq:15}
		\end{equation}
		
		The stationary group $ W $ of the point $ \overset {\circ} {k} $ is isomorphic to the group motions of $ M (3) $, 3-dimensional Euclidean space $ E_ {3} $. Each element $w \in \mathfrak{W}$ can be uniquely represented as:
		\begin{equation}
		w=R(\rho)B(\vec{z}\,),\label{eq:11}
		\end{equation}
		where $ \vec {z} $ is a 3-dimensional vector, and $ \rho $ is an orthogonal matrix $ 3 \times3: $
		\[
		\rho^{\intercal}\rho=I.
		\]
		The matrices $ R(\rho) $ and $ B(\vec{z}\,) $ have the following forms:
		\begin{equation}
		R(\rho)=\begin{pmatrix}1 & \vec{0\,}^{\intercal} & 0\\
		\vec{0} & \rho & \vec{0}\\
		0 & \vec{0\,}^{\intercal} & 1
		\end{pmatrix},\quad R(\rho)\in SO(3)\subset SO(4,1),\label{eq:12}
		\end{equation}
		\begin{equation}
		B(\vec{z})=\begin{pmatrix}1+z^{2}/2 & \vec{z\,}^{\intercal} & -z^{2}/2\\
		\vec{z} & I_{3} & -\vec{z}\\
		z^{2}/2 & \vec{z}\,^{\intercal} & 1-z^{2}/2
		\end{pmatrix},\quad z^{2}=\vec{z\,}^{2}.\label{eq:13}
		\end{equation}
		The matrices $ B (\vec {z}\,) $ form an abelian subgroup that is the normal subgroup of the stationary subgroup $ \mathfrak{W}. $
		
		For $ g \in SO (4,1) $ we have:
		\begin{equation}
		g=h(k)w.\label{eq:19}
		\end{equation}
		
		Hence, using the formulas \eqref{eq:15}-\eqref{eq:19}, we obtain the Iwasawa expansion for the elements of the group $ SO(4,1): $
		\begin{equation}
		g=ug_{04}(\beta)B(\vec{z}\,),\quad u\in SO(4).\label{eq:20}
		\end{equation}
		
		In doing so, we used the commutativity of the transformations $ g_ {04} (\beta) $ and $ R (\rho) $, and the fact that the group is locally $ SO (4) \simeq SO (3) \otimes SO(3) $.
		
		Here it is necessary to make two comments:
		\begin{enumerate}
			\item The action of the group $ SO(4,1) $ on the cone induces a nonlinear mapping section of the cone $ \omega = 1 $, i.e. 3-dimensional unit sphere \eqref{eq:14} to itself. It is easy to verify that this mapping is a homeomorphism. Thus, the 3-sphere is a homogeneous space, where the group $ SO (4,1) $ acts nonlinearly. For this space Wigner operator $ h (\upsilon) $ can be defined as follows:
			\begin{equation}
			h(\upsilon)=h(k)|_{\omega=1}.\label{eq:21}
			\end{equation}
			\item The mappings \eqref{eq:15} and \eqref{eq:21} are Borel. Moreover, they are differentiable almost everywhere.
		\end{enumerate}
		
		On the cone $ [k, k] = 0, \; k_ {0}> 0 $ we can introduce one more coordinate system:
		\begin{eqnarray}
		k_{0} & = & \tau\frac{1+a^{2}}{2},\nonumber \\
		\vec{k} & = & \tau\vec{a},\label{eq:22}\\
		k_{4} & = & \tau\frac{1-a^{2}}{2},\nonumber 
		\end{eqnarray}
		where $a=\vec{a\,}^{2},\quad\tau>0.$
		
		Since $ k_ {0} + k_ {4}> 0 $, this coordinate system does not cover the whole cone, namely, outside the coordinate system there remains the intersection of the cone by a hyperplane $ k_ {0} + k_ {4} = 0 $, which forms a set of smaller dimension. It is obvious that the corresponding parametrization of the group will also not be	global. The set of points $ k_ {0} + k_ {4}> 0 $ together with the map \eqref{eq:22} draws a map on the cone, which we denote by through $ \mathbb {C} _ {1}. $
		
		The coordinates of the fixed point $ \overset {\circ} {k} = (1/2, \vec {0}, 1/2) $:
		$ \tau = 1, \; \vec {a} = \vec {0}. $
		
		Wigner operator:
		\begin{equation}
		h(k)=A(\vec{a})D(\tau),\label{eq:23}
		\end{equation}
		where
		\begin{equation}
		D(\tau)=g_{04}(\alpha)=\begin{pmatrix}\cosh\alpha & \vec{0\,}^{\intercal} & \sinh\alpha\\
		\vec{0} & I_{3} & \vec{0}\\
		\sinh\alpha & \vec{0\,}^{\intercal} & \cosh\alpha
		\end{pmatrix},\qquad\tau=e^{\alpha}\label{eq:24}
		\end{equation}
		\begin{equation}
		A(\vec{a})=\begin{pmatrix}1+a^{2}/2 & \vec{a\,}^{\intercal} & a^{2}/2\\
		\vec{a} & I_{3} & \vec{a}\\
		-a^{2}/2 & -\vec{a\,}^{\intercal} & 1-a^{2}/2
		\end{pmatrix}\label{eq:25}
		\end{equation}
		
		The matrices $ A (\vec {a}\,) $, as well as the matrices $ B (\vec {z}\,) $ from \eqref{eq:13} form a 3-parameter abelian subgroup. Since, stationary subgroup is the group $\mathfrak{W}$ described in the previous subsection from \eqref{eq:11} and \eqref{eq:13} we get the expansion, which is valid \emph {almost for all} $ g \in SO(4,1): $
		\begin{equation}
		g=A(\vec{a}\,)D(\tau)R(r)B(\vec{b}\,)\label{eq:26}
		\end{equation}
		
		In order to obtain a parametrization for any element of the group, it is necessary to have an atlas on the cone without vertex.
		
		The next map of $ \mathbb {C} _ {2} $, consisting of the region $k_ {0} -k_ {4}> 0$  and mapping:
		\begin{eqnarray}
		k_{0} & = & \lambda\frac{1+x^{2}}{2},\label{eq:27}\\
		\vec{k} & = & \lambda\vec{x},\nonumber \\
		k_{4} & = & -\lambda\frac{1-x^{2}}{2};\nonumber \\
		x^{2} & = & \vec{x\,}^{2},\quad\lambda>0.\nonumber 
		\end{eqnarray}
		complements the map $ \mathbb {C} _ {1} $ to the atlas.
		
		Now we see that outside the coordinate system $\mathbb {C} {} _ {2}$
		remains the intersection of the cone with the hyperplane $k_{0}-k_{4} = 0$ , and that the maps $\mathbb {C} {}_{1}$   and $\mathbb {C} {} _ {2}$   make an atlas on the cone. The transition functions from one card to another are as follows:
		\begin{eqnarray}
		\tau & = & \lambda x^{2};\qquad\:\lambda=\tau a^{2},\label{eq:28}\\
		\vec{a} & = & \vec{x}/x^{2};\qquad\vec{x}=\vec{a}/a^{2}.\nonumber 
		\end{eqnarray}
		
		The formulas \eqref{eq:28} show that we have established a smooth structure on the cone.
		
		Now we need to construct the Wigner operator corresponding to the map \eqref{eq:27}. From \eqref{eq:6}-\eqref{eq:7} it follows that the matrix $ \eta \in SO(4,1) $, so you can enter the following boost:
		\begin{equation}
		h_{2}(\vec{k\,})=B(\vec{x\,})\eta D(\lambda).\label{eq:29}
		\end{equation}
		
		Thus, we obtain the following parametrization of $SO(4.1)$:
		\begin{equation}
		g=B(\vec{x\,})\eta D(\lambda)R(r)B(\vec{y\,}).\label{eq:30}
		\end{equation}
		
		It follows from the corresponding statements for the cone that the parameterizations \eqref{eq:26} and \eqref{eq:30} cover the entire group space. We also get that the expansion \eqref{eq:26} is valid for all $ g \in SO(4,1) $, except for the elements $ \left\{g \in SO(4,1): \: \vec {x} = 0 \right\} $, and \eqref{eq:30} is valid for all $ g \in SO (4,1) $, except for elements $ \left\{g \in SO(4,1): \: \vec {a} = 0 \right\} $.
		
		Here it is necessary to make a few remarks:
		\begin{enumerate}
			\item The elements $ \vec {a} = 0 $, i.e. elements of the form
			\begin{equation}
			p=D(\tau)R(r)B(\vec{b})\label{eq:31}
			\end{equation}
			constitute a subgroup, namely, a minimal parabolic subgroup $ \mathsf{P} $. This subgroup plays an important role in the construction of induced representations, \cite{Mensky}.
			
			Accordingly, the expansions \eqref{eq:26} and \eqref{eq:30} can now be rewritten as:
			\begin{equation}
			g=A(\vec{a})p_{1},\qquad\qquad p_{1}=D(\tau)R(r)B(\vec{b});\label{eq:33}
			\end{equation}
			\begin{equation}
			g=B(\vec{x})\eta p_{2};\qquad\qquad p_{2}=D(\lambda)R(r)B(\vec{y}).\label{eq:34}
			\end{equation}
			\item Below are the formulas for the transition between these parameterizations:
			\begin{eqnarray}
			p_{1} & = & D(x^{2})R\left(\pi;\overset{\circ}{\vec{x}}\right)B\left(-\frac{\vec{x}}{x^{2}}\right)p_{2};\qquad\qquad\vec{a}=\vec{x}/x^{2},\nonumber \\
			p_{2} & = & D(a^{2})R\left(\pi;\overset{\circ}{\vec{a}}\right)B\left(\frac{\vec{a}}{a^{2}}\right)p_{1};\qquad\qquad\;\vec{x}=\vec{a}/a^{2}.\label{eq:35}
			\end{eqnarray}
			
			Here, we denote by $ R \left (\pi; \cdot \right) $ the rotation by an angle
			$ \pi $ around the vectors $ \overset {\circ} {\vec {x}} = \vec {x} / x $ and $ \overset {\circ} {\vec {a}} = \vec {a} / a $, respectively.
			
			It is easy to see that $ \mathsf{P} $ is a stationary subgroup the north pole of the 3-sphere. Therefore:
			\[
			S^{3}\simeq SO(4,1)\diagup\mathsf{P}.
			\]
			\item Maps $ \mathbb {C}_{i}, \: i = 1,2 $ induce on a submanifold $\omega = 1$ cone, i.e. on the 3-dimensional unit sphere, the local maps $ \mathbb{V} _ {i}, \, i = 1,2 $, respectively. The map $ \mathbb{V}_{1} $ (respectively, $ \mathbb{V}_{2} $) consists of points of a sphere with a punctured southern (respectively, northern) pole and coordinate systems:
			\begin{equation}
			\upsilon_{\alpha}=\frac{2a_{\alpha}}{1+a^{2}},\qquad\qquad\upsilon_{4}=\frac{1-a^{2}}{1+a^{2}};\label{eq:36}
			\end{equation}
			respectively,
			\begin{equation}
			\upsilon_{\alpha}=\frac{2x_{\alpha}}{1+x^{2}},\qquad\qquad\upsilon_{4}=-\frac{1-x^{2}}{1+x^{2}}.\label{eq:37}
			\end{equation}
			
			As can be seen, these coordinates coincide with the stereographic projection 3-dimensional sphere from its poles to $ E_ {3}. $		
			\item The Wigner operators $ h_ {i} (\upsilon), \: i = 1,2 $ corresponding to the maps $ \mathbb {V} _ {i} $, will be defined as follows:
			\begin{equation}
			h_{1}(\upsilon)=A(\vec{a}),\qquad\qquad\qquad h_{2}(\upsilon)=B(\vec{x})\eta.\label{eq:38}
			\end{equation}
			\item The mappings \eqref{eq:23}, \eqref{eq:29} and \eqref{eq:38} are infinitely differentiable.
		\end{enumerate}
		
		In conclusion this section, we give some information about invariant measures on the group $ SO (4,1) $ and its minimally parabolic subgroup. The de Sitter group, like any connected semisimple group, unimodular, i.e. on it there exists a two-sided invariant measure Haar $ dg $. Below are the expressions for $ dg $ in different parameterizations (the numbers in parentheses on the left indicate the numbers of the corresponding parameterizations):
		\begin{eqnarray}
		\eqref{eq:20}: \ \quad\qquad\qquad\qquad e^{3\beta}dud\beta(d\vec{z\,}),\nonumber\\
		\eqref{eq:33}: \qquad\qquad\qquad\tau^{2}d\tau(d\vec{a})dr(d\vec{b\,}),\label{eq:39}\\
		\eqref{eq:34}: \qquad\qquad\qquad\lambda^{2}d\lambda(d\vec{x})dr(d\vec{y\,}),\nonumber	
		\end{eqnarray}
		
		where 
		
		$\qquad du $ is an invariant normalized measure on $ SO (4), $
		
		$\qquad dr $ is an invariant normed measure on $ SO (3); $
		
		$\qquad (d \vec {a\,}), \:( d \vec {b\,}), \:( d \vec {x\,}), \:( d \vec {y\,}) $ - volume elements in 3-dimensional Euclidean space.
		
		The subgroup $ \mathsf {P} $, in contrast to the group $ SO(4,1) $, is not unimodular. Below are the expressions for the right-invariant measures $ d_ {r}(p) $, the left-invariant measure $ d_ {l}(p) $, and also the module 	$ \delta (p) $:
		\begin{eqnarray}
		\delta(p) & = & \tau^{3},\nonumber \\
		d_{l}(p) & = & \frac{1}{\tau}d\tau dr\left(d\vec{b}\right),\label{eq:40}\\
		d_{r}(p) & = & \tau^{2}d\tau dr\left(d\vec{b}\right)=\delta(p)d_{l}(p).\nonumber 
		\end{eqnarray}
		
		Finally, from the comparison \eqref{eq:39} with \eqref{eq:40} we get:
		\begin{eqnarray}
		dg & = & dh_{i}d_{r}(p),\qquad i=1,2;\nonumber \\
		dh_{1} & = & \left(d\vec{a}\right),\label{eq:42}\\
		dh_{2} & = & \left(d\vec{x}\right).\nonumber 
		\end{eqnarray}

\end{document}